\documentclass[a4paper,11pt,fleqn]{article}

\usepackage{amsmath,amssymb,amsthm}%,backref
\usepackage[mathscr]{eucal}
\usepackage{hyperref}
\hypersetup{colorlinks, linkcolor=blue, citecolor=blue, urlcolor=blue}

\newcommand{\todo}[1][\null]{\ensuremath{\clubsuit}}

\newcommand{\noprint}[1]{}
\newcommand{\checked}[1][\null]{\ensuremath{\boldsymbol{\surd}}}

\newcommand{\p}{\partial}
\newcommand{\ord}{\mathop{\rm ord}\nolimits}
\newcommand{\lsemioplus}{\mathbin{\mbox{$\lefteqn{\hspace{.77ex}\rule{.4pt}{1.2ex}}{\in}$}}}

\newtheorem{theorem}{Theorem}
\newtheorem{lemma}[theorem]{Lemma}
\newtheorem{corollary}[theorem]{Corollary}
\newtheorem{proposition}[theorem]{Proposition}
\newtheorem*{problem*}{Problem}
{\theoremstyle{definition}

\newtheorem{remark}[theorem]{Remark}
\newtheorem*{remark*}{Remark}
}

\begin{document}

\par\noindent {\LARGE\bf
Generalized symmetries and conservation laws\\ of (1+1)-dimensional Klein--Gordon equation
\par}

\vspace{4mm}\par\noindent {\large Stanislav Opanasenko$^{\dag\S}$ and Roman O.\ Popovych$^{\ddag\S}$
} \par\vspace{2mm}\par

\vspace{2mm}\par\noindent {\it
$^{\dag}$~Department of Mathematics and Statistics, Memorial University of Newfoundland,\\
$\phantom{^{\dag}}$~St.\ John's (NL) A1C 5S7, Canada\par
}
\vspace{2mm}\par\noindent {\it
$^\ddag$~Fakult\"at f\"ur Mathematik, Universit\"at Wien, Oskar-Morgenstern-Platz 1, 1090 Wien, Austria
$\phantom{^\S}$~\& Mathematical Institute, Silesian University in Opava, Na Rybn\'\i{}\v{c}ku 1, 746 01 Opava,\\
$\phantom{^\S}$~Czech Republic\par
}
\vspace{2mm}\par\noindent {\it
$^\S$~Institute of Mathematics of NAS of Ukraine, 3 Tereshchenkivs'ka Str., 01004 Kyiv, Ukraine

}

\noprint{

35-XX   Partial differential equations
 35Bxx  Qualitative properties of solutions
  35B06   Symmetries, invariants, etc.
37-XX   Dynamical systems and ergodic theory [See also 26A18, 28Dxx, 34Cxx, 34Dxx, 35Bxx, 46Lxx, 58Jxx, 70-XX]
 37Kxx  Infinite-dimensional Hamiltonian systems [See also 35Axx, 35Qxx]
  37K05   Hamiltonian structures, symmetries, variational principles, conservation laws
35-XX	Partial differential equations
 35Lxx	  Hyperbolic equations and systems
  35L10  	Second-order hyperbolic equations
  35L65  	Conservation laws
}

\vspace{2mm}\par\noindent
\textup{E-mail:} sopanasenko@mun.ca, rop@imath.kiev.ua
\par

\vspace{8mm}\par\noindent\hspace*{10mm}\parbox{140mm}{\small
Using advantages of nonstandard computational techniques based on the light-cone variables,
we explicitly find the algebra of generalized symmetries of the (1+1)-dimensional Klein--Gordon equation.
This allows us to describe this algebra in terms of the universal enveloping algebra
of the essential Lie invariance algebra of the Klein--Gordon equation.
Then we single out variational symmetries of the corresponding Lagrangian
and compute the space of local conservation laws of this equation,
which turns out to be generated, up to the action of generalized symmetries,
by a single first-order conservation law.
Moreover, for every conservation law we find a conserved current of minimal order
that is contained in this conservation law.
\looseness=-1
\par}\vspace{4mm}

\section{Introduction}

%\looseness=-1
Noether's idea of generalizing the notion of Lie symmetries of systems of differential equations
was to allow components of vector fields to depend on derivatives of unknown functions,
which led to the notion of generalized (or higher) symmetries \cite{Bocharov1999,Olver1993}.
In this way, symmetries lose their geometric charm but become a powerful tool, e.g.,
for finding, with Noether's theorem, conservation laws of systems that are systems of Euler--Lagrange equations for some Lagrangians.
Although the general procedure of finding generalized symmetries is similar to its counterpart for Lie symmetries, 
computational difficulty increases rapidly as the order of symmetries to be found increases. 
Even low-order generalized symmetries are hard to compute for multidimensional systems of differential equations,
in spite of the possibility of using specialized computer algebra packages~\cite{BaranMarvan,Cheviakov2007} in such computations.
The situation with (local) conservation laws is alike,
see for instance remarks in~\cite{CheviakovOberlack2014} on computational complexity of the problem
on conservation laws of the Euler and the Navier--Stokes equations of order less than or equal to two.
Besides, given a system of differential equations,
a computer cannot handle the construction of all generalized symmetries or conservation laws of this system
unless there exist upper bounds on their orders, and these bounds are quite low and are found independently.
In view of this, the complete descriptions of generalized symmetries and/or of conservation laws are known
for not so many systems of differential equations important for real-world applications
as may be expected, taking into account the intensive research activity in the related~field.

%\looseness=-1
The above approach with computing the upper bound of orders of generalized symmetries, cosymmetries or conservation laws
was applied for a number of systems of differential equations for which such bounds exist.
This includes conservation laws of the BBM equations~\cite{DuzhinTsujishita1984,Olver1979b},
of the $k$-$\varepsilon$ turbulence model~\cite{Khor'kovaVerbovetsky1995},
of (1+1)-dimensional even-order linear evolution equations~\cite[Corollary~6]{PopovychSergyeyev2010}
and of the equation $u_t=u_{xxx}+xu$~\cite[Example~6]{PopovychSergyeyev2010},
the classification of conservation laws of second-order evolution equations~\cite{PopovychSamoilenko2008} up to contact equivalence,
generalized symmetries of the Bakirov system~\cite{Sergyeyev2001} 
as well as
generalized symmetries and conservation laws 
of the Navier--Stokes equations~\cite{GusyatnikovaYumaguzhin1989},
of the (1+3)-dimensional, (1+2)-dimensional and axisymmetric Khokhlov--Zabolotskaya equations~\cite{Sharomet1989},
of non-integrable compacton $K(m,m)$-equations~\cite{Vodova2013}
and of generalized Kawahara equations~\cite{Vasicek2020}.
There exist no more or less general results on such upper bounds,
except the well-known upper bound for orders of conservation laws of even-order \mbox{(1+1)-dimensional} evolution equations
and the extension of this bound in~\cite{Igonin2002} to a wider class of systems of differential equations.

For (integrable) systems admitting (co)symmetries of arbitrary high order,
it may be possible to find recursion operators~\cite{KrasilshchikVerbovetskyVitolo2017,Olver1977,Olver1993,Sergyeyev2017}
for symmetries and/or for cosymmetries with subsequent determining which cosymmetries are associated with conservation laws.
At the same time, recursion operators are not guaranteed to yield all (co)symmetries
and so there remains a problem of proving nonexistence of other (co)symmetries.
Another point is that recursion operators do not always generate local objects,
with generalized symmetries of the Korteweg--de Vries equation and the Lenard recursion operator~\cite{GardnerGreeneKruskalMiura1974}
as an example here, so it is necessary to pick the local ones post factum
or prove that the generated hierarchy is local~\cite{Sergyeyev2005}.
Amongst known examples of complete descriptions of generalized symmetries and conservation laws
for systems admitting such objects of arbitrarily high order are those 
for the Korteweg--de Vries equation~\cite{KruskalMiuraGardnerZabusky1970,IbragimovShabat1979,MagadeevSokolov1981,Tsujishita1982},
for its linear counterpart $u_t=u_{xxx}$~\cite[Example~5]{PopovychSergyeyev2010}, 
of the vacuum Einstein equations in the four-dimensional spacetime~\cite{AndersonTorre1996},
for free Maxwell's equations in (3+1)-dimensional Minkowski space~\cite{AncoPohjanpelto2001,AncoPohjanpelto2004}, % in classical electromagnetic field theory
for massless free fields of spin $s\geqslant1/2$ \cite{AncoPohjanpelto2003,PohjanpeltoAnco2008} and 
for an isothermal no-slip drift flux model~\cite{OpanasenkoBihloPopovychSergyeyev2020b}. 
All the generalized symmetries of the Yang--Mills equations on Minkowski space with a semi-simple structure group 
were computed in~\cite{Pohjanpelto2004}.
Symmetry operators of the one-dimensional Schr\"odinger equation were studied in~\cite{FushchichNikitin1994,NikitinOnufriichukFuschchich1992}.
See also~\cite{DubrovinNovikov1989,Tsarev1991} for a general theory of hydrodynamic systems,
where infinite hierarchies of conservation laws and symmetries though often nonlocal are common,
and~\cite{MorozovSergyeyev2014,Sergyeyev2017,Sergyeyev2018,Vinogradov1989} for some related examples.

In the present paper, we exhaustively describe generalized symmetries and local conservation laws
of the (1+1)-dimensional (real) Klein--Gordon equation, which takes, in natural units, the~form
\[
\square u+m^2u=0,
\]
where $u$ is the real-valued unknown function of the real independent variables~$x_0$ and~$x_1$,
$\square$~is the d'Alembert operator in (1+1) dimensions, $\square=\p^2/\p x_0^2-\p^2/\p x_1^2$,
and $m$ denotes the nonzero mass parameter.%
\footnote{%
The zero value of~$m$, which corresponds to the wave equation, is singular in all properties
related to symmetry analysis of differential equations,
including Lie, contact and generalized symmetries and conservation laws; 
cf.\ \cite[Section~18.4]{Ibragimov1985} and~\cite{PopovychCheviakov2020}.
}
Without loss of generality, the mass parameter can be set to be equal one
by simultaneous scaling of the independent variables.
We work with this equation in the characteristic, or light-cone, variables
(the change of variables $x=(x_0+x_1)/2$, $y=(-x_0+x_1)/2$ does the job),
\begin{gather*}%\label{eq:(1+1)DKGEq}
\mathcal L\colon\quad u_{xy}=u.
\end{gather*}
In what follows we use the same notation~$\mathcal L$ for the solution set of the equation~$\mathcal L$
as well as for the set defined by~$\mathcal L$ and its differential consequences in the corresponding infinite-order jet space.
Our specific interest to the equation~$\mathcal L$ originated from the study of
the hydrodynamic-type system~$\mathcal S$ of differential equations
modeling an isothermal no-slip drift flux \cite{OpanasenkoBihloPopovychSergyeyev2020a,OpanasenkoBihloPopovychSergyeyev2020b}.
It turned out that the (nonlinear) system~$\mathcal S$ is reduced to the (linear) equation~$\mathcal L$
by the composition of a simple point transformation
and a rank-two hodograph transformation.
The family of regular solutions of~$\mathcal S$
is parameterized by an arbitrary solution of~$\mathcal L$ and by an arbitrary function of a single argument.
Moreover, finding generalized symmetries and local conservation laws of the system~$\mathcal S$
reduces to the analogous problems for the equation~$\mathcal L$.
At the same time, we did not find exhaustive and trusted solutions of the latter problems in the literature,
which motivated our study of the Klein--Gordon equation.

The Lie invariance algebra~$\mathfrak g$ of the equation~$\mathcal L$ was computed by Sophus Lie himself
in the course of the group classification of second-order linear equations with two independent variables~\cite[Section~9]{Lie1881}.
The equation~$\mathcal L$ appeared there
as the simplest particular member of a parameterized family of inequivalent equations
that admit three-dimensional Lie-symmetry extensions in comparison with the general case.%
\footnote{%
The same classification case was represented in \cite[Section~9.6]{Ovsiannikov1982} by
another family, which is similar to the family singled out by Lie
with respect to a point transformation but is more cumbersome.
Under this representation, the relation of the Klein--Gordon equation to Lie-symmetry extensions
within the class of second-order linear equations with two independent variables
is not so obvious as in Lie's paper~\cite{Lie1881}.
}
The algebra~$\mathfrak g$ is spanned by the vector fields
\[
\p_x,\ \p_y,\ x\p_x-y\p_y,\ u\p_u,\ f(x,y)\p_u,
\]
where the function $f=f(x,y)$ runs through the solution set of~$\mathcal L$.
This algebra is represented as the semidirect sum,
$\mathfrak g=\mathfrak g^{\rm ess}\lsemioplus\mathfrak g^\infty$,
of the so-called (finite-dimensional) essential Lie invariance subalgebra
$\mathfrak g^{\rm ess}:=\langle\p_x,\,\p_y,\,x\p_x-y\p_y,\,u\p_u\rangle$
and the (infinite-dimensional) Abelian ideal
\mbox{$\mathfrak g^\infty:=\langle f(x,y)\p_u,\,f\in\mathcal L\rangle$}
related to the linear superposition of solutions of~$\mathcal L$.
Note that Sophus Lie carried out the group classification over the complex field under supposing
all objects, like equation coefficients and components of vector fields, to be analytic.
This is why his results are directly extended to hyperbolic equations over the real field.

Since the equation~$\mathcal L$ is the Euler--Lagrange equation of the Lagrangian $\mathrm L=-(u_xu_y+u^2)/2$,
its local conservation laws can be constructed using Noether's theorem.
Conservation laws associated with essential variational Lie symmetries of the Lagrangian~$\mathrm L$ are well known
and admit an obvious physical interpretation.
These are the conservations of energy-momentum and of relativistic angular momentum,
which are respectively related, via Noether's theorem,
to spacetime translations and to Lorentz transformations, which include, in dimension 1+1, only Lorentz boosts;
see~\cite{Torre2016} for a nice pedagogical presentation.

In the course of a general discussion of quadratic conserved quantities in free-field theories in~\cite{Kibble1965},
it was shown that the (1+3)-dimensional Klein--Gordon equation possesses
an infinite-dimensional space of conservation laws with conserved currents
whose components are quadratic expressions in derivatives of the dependent variable with constant coefficients;
in fact, the specific dimension (1+3) is not essential in this result.
Tsujishita~\cite{Tsujishita1979} proved that for the $(1+n)$-dimensional Klein--Gordon equation with $n\geqslant 2$,
this space coincides with the space of conservation laws containing the conserved currents
whose components are differential polynomials with constant coefficients;
see also~\cite{Tsujishita1982} and references therein.
At the same time, the Klein--Gordon equation obviously possesses other conservation laws.
There are such conservation laws even among conservation laws associated with Lie variational symmetries of the corresponding Lagrangian,
e.g., the conservations of relativistic angular momentum.

Having generalized the notion of Killing vector, in~\cite{Nikitin1991}
Nikitin introduced the notions of generalized Killing tensors and generalized conformal Killing tensors
of arbitrary rank and arbitrary order
in the $(p{+}q)$-dimensional pseudo-Euclidean space~$\mathbb R^{p,q}$ of signature $(p,q)$
with arbitrary $p,q\in\mathbb N_0:=\mathbb N\cup\{0\}$, $p+q>1$.
The explicit form of these tensors was found therein and
then used for the study of linear symmetry operators of the Klein--Gordon--Fock equation in~$\mathbb R^{p,q}$.
See also~\cite{NikitinPrylypko1990} for a more detailed exposition of the above results
and~\cite{FushchichNikitin1994}, where a number of results on linear symmetry operators
of linear systems of differential equations arising as models in quantum mechanics are collected.

Shapovalov and Shirokov stated in~\cite{ShapovalovShirokov1992} that for any $r\in\mathbb N_0$,
an arbitrary linear second-order partial differential equation with nondegenerate symbol and more than two independent variables
possesses only a finite number of linearly independent linear symmetry operators up to order~$r$
and admits no ``nonlinear'' generalized symmetries.
Therein, they also described the algebra of generalized symmetries of the Laplace--Beltrami equation in the space~$\mathbb R^{p,q}$
in terms of the universal enveloping algebra of the essential Lie invariance algebra of this equation;
see~\cite{Eastwood2005} for further deeper study of the algebra of generalized symmetries of the Laplace equation.

Note that the algebra of generalized symmetries 
and the spaces of local conservation laws and variational symmetries of the associated Lagrangian
of the allied (1+1)-dimensional wave equation $u_{xy}=0$
are known, see~\cite[Section~18.4]{Ibragimov1985} and~\cite{PopovychCheviakov2020},
and they essentially differ from the corresponding objects for the equation~$\mathcal L$.
Nonlinear wave equations of the form \mbox{$u_{xy}=f(u)$} admitting generalized symmetries
whose characteristics do not depend on the independent variables
were singled out in~\cite{ZhiberShabat1979}; see also \cite[Section~21.2]{Ibragimov1985}.
The complete classification of local conservation laws of equations in this class was initiated
and partially carried out in~\cite{FoxGoertsches2011}.

The structure of the present paper is as follows.
In Section~\ref{KG:sec:GenSyms} we explicitly describe the quotient algebra~$\Sigma^{\rm q}$ of generalized symmetries
of the (1+1)-dimensional Klein--Gordon equation~$\mathcal L$ with respect to the standard equivalence of generalized symmetries
by presenting a naturally isomorphic space of representatives for equivalence classes of generalized symmetries.
This leads to the description of the algebra~$\Sigma^{\rm q}$ in terms of the universal enveloping algebra
of the essential Lie invariance algebra of~$\mathcal L$.
The related computations are essentially simplified by using advantages of the characteristic independent variables for the equation~$\mathcal L$,
which are specific for the (1+1)-dimensional case.
As another optimization, we avoid the direct integration of the system of determining equations for generalized symmetries of~$\mathcal L$.
Instead of this integration, which is realizable but quite cumbersome,
we estimate the number of independent linear symmetries of an arbitrary fixed order,
apply the Shapovalov--Shirokov theorem~\cite{ShapovalovShirokov1992}
and explicitly present the same number of appropriate linear symmetries.
In Section~\ref{KG:sec:VarSyms}
we recall the variational interpretation of the equation~$\mathcal L$ and
accurately single out the space of variational symmetries of the Lagrangian~$\mathrm L$
from the entire space of generalized symmetries of~$\mathcal L$.
Finally, in Section~\ref{KG:sec:ConsLaws} we find the space of local conservation laws of~$\mathcal L$ using Noether's theorem for
constructing a space of conserved currents that is naturally isomorphic to the space of local conservation laws.
In the course of this construction, we select conserved currents of minimal order among the equivalent ones,
which immediately specifies the spaces of conservation laws of each fixed order.
We also show that, up to the action of generalized symmetries, the entire space of conservation laws of the equation under
study is generated by a single conservation law.
In Section~\ref{KG:sec:Conclusion} we underscore all the techniques and ideas, especially specific
to the present paper, which we use in the course of the study.
\looseness=-1

\section{Generalized symmetries}\label{KG:sec:GenSyms}

Here we revisit the construction of the algebra~$\Sigma$ of generalized symmetries of the
(1+1)-dimensional Klein--Gordon equation with some enhancements.
Computing generalized symmetries,
without loss of generality we can consider only evolutionary generalized vector fields
and evolutionary representatives of generalized symmetries~\cite[p.~291]{Olver1993}
and thus assume that the algebra~$\Sigma$ is constituted by such representatives for the above equation,
\[
\Sigma=\big\{Q=\eta[u]\p_u\mid\mathrm D_x\mathrm D_y\eta[u]=\eta[u] \mbox{ on } \mathcal L\big\},
\]
where $\eta[u]$ denotes a differential function of~$u$,
and $\mathrm D_x$ and $\mathrm D_y$ are the operators of total derivatives in~$x$ and~$y$, respectively;
see~\cite[Definition~2.34]{Olver1993}.
We denote by $\Sigma^{\rm triv}$ the algebra of trivial generalized symmetries of the equation~$\mathcal L$,
which is an ideal of~$\Sigma$.
It consists of all generalized vector fields in the evolutionary form
(with the independent variables $(x,y)$ and the dependent variable~$u$)
whose characteristics vanish on solutions of~$\mathcal L$.
The quotient algebra $\Sigma^{\rm q}=\Sigma/\Sigma^{\rm triv}$ is naturally isomorphic%
\footnote{%
There are two similar kinds of \emph{natural ({\rm or} canonical) isomorphisms} in this paper---%
those related to quotient linear spaces and those related to quotient Lie algebras.
Given a linear space~$V$ and its subspaces~$U$ and~$W$ such that $V=U\dotplus W$,
where ``$\dotplus$'' denotes the direct sum of subspaces,
the natural isomorphism between $V/U$ and $W$ is established in the way that
each coset of~$U$ corresponds to the unique element of~$W$ belonging to this coset.
In a similar way, natural isomorphisms are established between $\mathfrak a/\mathfrak i$ and $\mathfrak b$,
where $\mathfrak a$ is a Lie algebra, and $\mathfrak b$ and~$\mathfrak i$ are its subalgebra and its ideal, respectively,
such that $\mathfrak a=\mathfrak b\lsemioplus\mathfrak i$.%\\[-1ex]
}
to the algebra of canonical representatives in the reduced evolutionary form,\looseness=-1
\[
\hat\Sigma^{\rm q}=\big\{Q=\eta[u]\p_u\in\Sigma\mid\eta[u]=\eta(x,y,u_{-n},\dots,u_n) \mbox{ for some } n\in\mathbb N_0\big\}.
\]
Here $x$, $y$, $u_0:=u$, $u_k:=\p_x^ku$ and $u_{-k}:=\p_y^ku$, $k\in\mathbb N$,
constitute the standard coordinates on the manifold defined by the equation~$\mathcal L$
and its differential consequences in the infinite-order jet space $\mathrm J^\infty(x,y|u)$
with the independent variables $(x,y)$ and the dependent variable~$u$. 
The notation $u_{-k}$ is justified in view of the operator equality 
$\p_x|_{\mathcal L}^{}\circ\p_y|_{\mathcal L}^{}=\mathrm{id}|_{\mathcal L}^{}$,  
where $|_{\mathcal L}^{}$ denotes the restriction to the solution set of~$\mathcal L$.
The Lie bracket on~$\hat\Sigma^{\rm q}$ is defined
as the reduced Lie bracket of generalized vector fields, where all arising mixed derivatives of~$u$ are
substituted in view of the equation~$\mathcal L$ and its differential consequences,
\[
[\eta^1\p_u,\eta^2\p_u]
=\sum_{k=0}^{\infty}(\eta^2_{u_k}\mathscr D_x^k\eta^1-\eta^1_{u_k}\mathscr D_x^k\eta^2)\p_u
+\sum_{k=1}^{\infty}(\eta^2_{u_{-k}}\mathscr D_y^k\eta^1-\eta^1_{u_{-k}}\mathscr D_y^k\eta^2)\p_u,
\]
where $\mathscr D_x$ and~$\mathscr D_y$ are the reduced operators of total derivatives with respect to~$x$ and~$y$, respectively,
\[
\mathscr D_x:=\p_x+\sum_{k=-\infty}^{+\infty}u_{k+1}\p_{u_k},\quad
\mathscr D_y:=\p_y+\sum_{k=-\infty}^{+\infty}u_{k-1}\p_{u_k}.
\]

The subspace
\[
\Sigma^n=\big\{[Q]\in\Sigma^{\rm q}\mid\exists\,\eta[u]\p_u\in[Q]\colon\ord\eta[u]\leqslant n\big\}, \quad n\in\mathbb N_0\cup\{-\infty\},
\]
of~$\Sigma^{\rm q}$ is interpreted as the space of generalized symmetries of order less than or equal to~$n$.%
\footnote{
The order~$\ord F[u]$ of the differential function~$F[u]$ is the highest order of derivatives of~$u$ involved in~$F[u]$
if there are such derivatives, and $\ord F[u]=-\infty$ otherwise.
If $Q=\eta[u]\p_u$, then $\ord Q:=\ord\eta[u]$.
For $[Q]\in\Sigma^{\rm q}$, $\ord[Q]=\min\big\{\ord\eta[u]\mid\eta[u]\p_u\in[Q]\big\}$.
}
It is naturally isomorphic to the subspace of canonical representatives in the reduced evolutionary form
with characteristics of order less than or equal to~$n$,
\[
\hat\Sigma^n=\big\{\eta[u]\p_u\in\hat\Sigma^{\rm q}\mid\ord\eta[u]\leqslant n\big\},\quad n\in\mathbb N_0\cup\{-\infty\}.
\]
Note that the subspace~$\hat\Sigma^{-\infty}$ can be identified with the subalgebra of Lie symmetries of~$\mathcal L$
associated with the linear superposition of solutions of~$\mathcal L$,
\[\hat\Sigma^{-\infty}=\{f(x,y)\p_u\mid f\in\mathcal L\}.\] 
The notation $f\in\mathcal L$ means that the function~$f$ runs through the solution set of~$\mathcal L$.
The subspace family $\{\Sigma^n\mid n\in\mathbb N_0\cup\{-\infty\}\}$ filters the algebra~$\Sigma^{\rm q}$.
Consider the quotient spaces \mbox{$\Sigma^{[n]}=\Sigma^n/\Sigma^{n-1}$} for $n\in\mathbb N$ and
$\Sigma^{[0]}=\Sigma^0/\Sigma^{-\infty}$ and denote $\Sigma^{[-\infty]}:=\Sigma^{-\infty}$.
The space $\Sigma^{[n]}$ can be assumed as the space of $n$th order generalized symmetries of~$\mathcal L$, $n\in\mathbb N_0\cup\{-\infty\}$.

Since the equation~$\mathcal L$ is linear,
an important subalgebra of its generalized symmetries consists of linear generalized symmetries,
\[
\Lambda=\bigg\{\eta[u]\p_u\in\Sigma\ \Big|\ \eta=\mathscr Du
\mbox{ \ for some \ } \smash{\mathscr D=\sum_{|\alpha|\leqslant n}\zeta^\alpha(x,y)\mathrm D_x^{\alpha_1}\mathrm D_y^{\alpha_2}},\ n\in\mathbb N_0\bigg\}.
\]
Recall that $\alpha=(\alpha_1,\alpha_2)\in\mathbb N_0^{\,\,2}$ is a multiindex, and $|\alpha|=\alpha_1+\alpha_2$.
The subalgebra~$\Lambda^{\rm triv}$ of trivial linear generalized symmetries coincides with $\Lambda\cap\Sigma^{\rm triv}$.
The quotient algebra $\Lambda^{\rm q}=\Lambda/\Lambda^{\rm triv}$ can be embedded into~$\Sigma^{\rm q}$
as the subalgebra of cosets of~$\Sigma^{\rm triv}$ that contain linear generalized symmetries.
The subspace $\Lambda^n=\Lambda^{\rm q}\cap\Sigma^n$ with $n\in\mathbb N_0$
is naturally isomorphic to the space~$\hat\Lambda^n$ of evolutionary generalized symmetries whose characteristics are of the reduced form,
where the mixed derivatives of~$u$ are excluded in view of~$\mathcal L$,
\begin{gather}\label{eq:ReducedLinGenSymsOf(1+1)DKGEq}
\eta[u]=\sum_{k=-n}^n\eta^k(x,y)u_k.
\end{gather}
Elements of~$\hat\Lambda^n$ are canonical representatives of cosets of~$\Sigma^{\rm triv}$ constituting the space~$\Lambda^n$.
The quotient spaces $\Lambda^{[n]}=\Lambda^n/\Lambda^{n-1}$, $n\in\mathbb N$, and the subspace $\Lambda^{[0]}=\Lambda^0$
are naturally embedded into the respective spaces $\Sigma^{[n]}$'s, $n\in\mathbb N_0$.
We interpret the space $\Lambda^{[n]}$ as the space of $n$th order linear generalized symmetries of~$\mathcal L$, $n\in\mathbb N_0$.
This space is isomorphic to the space of the pairs $(\eta^n,\eta^{-n})$
such that the differential function~$\eta[u]$ defined by~\eqref{eq:ReducedLinGenSymsOf(1+1)DKGEq}
with some values of the other coefficients~$\eta$'s
is the characteristic of an element of~$\hat\Lambda^n$.

\begin{lemma}\label{lemma:IDFM:DimLambda}
$\dim\Lambda^{[n]}=2n+1$, $n\in\mathbb N_0$.
\end{lemma}

\begin{proof}
For linear generalized symmetries with characteristics of the form~\eqref{eq:ReducedLinGenSymsOf(1+1)DKGEq},
the invariance criterion for the equation~$\mathcal L$, $\mathscr D_x\mathscr D_y\eta=\eta$,
implies the following system of determining equations:
\begin{gather}\label{eq:DetEqsForReducedLinGenSymsOf(1+1)DKGEq}
\Delta_k\colon\ \eta^k_{xy}+\eta^{k-1}_y+\eta^{k+1}_x=0,\quad k=-n-1,-n,\dots,n,n+1,
\end{gather}
where we assume $\eta^{-n-2}$, $\eta^{-n-1}$, $\eta^{n+1}$ and $\eta^{n+2}$ to vanish.
These symmetries are of (essential) order~$n$ if and only if
at least one of the coefficients~$\eta^{-n}$ and $\eta^n$ does not vanish.

Suppose that the coefficient~$\eta^{-n}$ does not vanish.
We integrate the equation~$\Delta_{-n-1}$: $\eta^{-n}_x=0$, which gives $\eta^{-n}=\theta(y)$ for some smooth function~$\theta$ of~$y$.
After substituting the obtained value of~$\eta^{-n}$ into $\Delta_{-n}$ and $\Delta_{-n+1}$,
we consider the set~$\Delta_{[-n,n-1]}$ of the equations~$\Delta_k$ with $k=-n,-n+1,\dots,n-1$
as a system of inhomogeneous linear differential equations with respect to the other $\eta$'s.
The equation $\Delta_{-n}$ takes the form $\eta^{-n+1}_x=0$,
and it is convenient to represent the equations~$\Delta_k$ with $k=-n+1,-n+2,\dots,n-1$ as
$\eta^{k+1}_x=-\eta^k_{xy}-\eta^{k-1}_y$.
To find a particular solution of the system~$\Delta_{[-n,n-1]}$,
we successively integrate its equations with respect to~$x$,
taking the antiderivatives $0$ and $x^{n+1}/(n+1)$ for $0$ and~$x^n$, respectively.
We can neglect the solutions of the homogeneous counterpart of~$\Delta_{[-n,n-1]}$
since they correspond to the zero value of~$\eta^{-n}$.
After the integration, we derive an expression for $\eta^n$ of the form
\[
\eta^n=\frac{(-1)^n}{n!}\frac{{\rm d}^n\theta}{{\rm d}y^n}x^n+R,
\]
where $R$ is a polynomial in~$x$ with $\deg_xR<n$,
whose coefficients depend linearly and homogeneously on derivatives of~$\theta$ of order greater than~$n$.
Substituting this expression into the equation~$\Delta_n$: $\eta^n_y=0$
and splitting with respect to~$x$,
we obtain the equation ${\rm d}^{n+1}\theta/{\rm d}y^{n+1}=0$.
Since the derivative~$\eta^{n-1}_y$ is of the same structure as~$R$,
the equation~$\Delta_{n-1}$: $\eta^{n-1}_y=0$ is identically satisfied in view of the equation for~$\theta$.
As a result, we have $n+1$ linearly independent values of the coefficient~$\eta^{-n}$, say, $1$,~$y$,~\dots,~$y^n$,
and, therefore, $n+1$ linearly independent generalized symmetries with characteristics
of the form~\eqref{eq:ReducedLinGenSymsOf(1+1)DKGEq} with nonvanishing coefficient~$\eta^{-n}$.
Moreover, only one of these symmetries, with $\eta^{-n}=y^n$, has a nonvanishing value of the coefficient~$\eta^n$.

Since the problem is symmetric with respect to~$x$ and~$y$,
after supposing that the coefficient~$\eta^n$ does not vanish,
we turn the above procedure around by permuting~$x$ and~$y$ and by changing the direction of the successive integration.
This leads to $n+1$ linearly independent generalized symmetries with characteristics
of the form~\eqref{eq:ReducedLinGenSymsOf(1+1)DKGEq} with nonvanishing coefficient~$\eta^n$,
where similarly to the above case, only one of these symmetries has a nonvanishing value of the coefficient~$\eta^{-n}$.

Therefore, in total there exist precisely $2n+1$ linearly independent $n$th order generalized symmetries
with characteristics of the form~\eqref{eq:ReducedLinGenSymsOf(1+1)DKGEq}.
\end{proof}

\begin{corollary}
$\dim\Lambda^n=\sum\limits_{k=0}^n\dim\Lambda^{[k]}=(n+1)^2<+\infty$, $n\in\mathbb N_0$.
\end{corollary}

\begin{lemma}\label{lemma:IDFM:SigmaIsomorphism}
The space~$\Sigma^{[n]}$ with $n\in\mathbb N_0$ is naturally isomorphic to the subspace
\[
\tilde{\Sigma}^{[n]}=\big\langle(\mathrm J^nu)\p_u,\,(\mathrm J^k\mathrm D_x^{n-k}u)\p_u,\,(\mathrm J^k\mathrm D_y^{n-k}u)\p_u,\,k=0,\dots,n-1\big\rangle
\]
of~$\Lambda$, where $\mathrm J:=x\mathrm D_x-y\mathrm D_y$.
Here each element~$Q$ of~$\tilde{\Sigma}^{[n]}$ corresponds to the element of~$\Sigma^{[n]}$
that, as a coset of~$\Sigma^{n-1}$ in~$\Sigma^n$, contains an element of $\Sigma^n$
that, as a coset of~$\Sigma^{\rm triv}$ in~$\Sigma$, contains~$Q$.
\end{lemma}

\begin{proof}
In view of the Shapovalov--Shirokov theorem~\cite[Theorem~4.1]{ShapovalovShirokov1992}, Lemma~\ref{lemma:IDFM:DimLambda}
implies that $\Sigma^{[n]}=\Lambda^{[n]}$ for $n\in\mathbb N_0$.

The differential functions $\mathrm D_xu=u_x$, $\mathrm D_yu=u_y$ and~$\mathrm Ju=xu_x-yu_y$ are the characteristics
of the Lie symmetries~$-\p_x$, $-\p_y$ and~$y\p_y-x\p_x$ of the equation~$\mathcal L$, respectively,
and hence the operators~$\mathrm D_x$, $\mathrm D_y$ and~$\mathrm J$ are recursion operators of~$\mathcal L$.
Therefore, any operator~$\mathfrak D$ in the universal enveloping algebra generated by these operators
is a symmetry operator of~$\mathcal L$,
that is, a generalized vector field~$(\mathfrak Du)\p_u$ is a generalized symmetry of~$\mathcal L$.
Thus, $\tilde\Sigma^{[n]}\subset\Lambda\subset\Sigma$.

The space~$\tilde{\Sigma}^{[n]}$ contains no nonzero trivial generalized symmetries of~$\mathcal L$.
Indeed, suppose that an element~$Q\in\tilde{\Sigma}^{[n]}$ with characteristic
\[
Q[u]=a\mathrm J^nu+\sum_{k=0}^{n-1}\big(b_k\mathrm J^k\mathrm D_x^{n-k}u+c_k\mathrm J^k\mathrm D_y^{n-k}u\big)
\]
is a trivial symmetry, that is, $Q[u]$ vanishes on solutions of~$\mathcal L$. Here~$a$, $b$'s and~$c$'s are constants.
Consider the solution~$u^\lambda={\rm e}^{\lambda x+\lambda^{-1}y}$ of the equation~$\mathcal L$, which is parameterized
by~$\lambda\in\mathbb R/\{0\}$.
The expression ${\rm e}^{-\lambda x-\lambda^{-1}y}Q[u^\lambda]$ is a polynomial in $\lambda x-\lambda^{-1}y$, $\lambda x+\lambda^{-1}y$, $\lambda$ and  $\lambda^{-1}$,
whose collection of terms of maximal total degree, which equals~$n$, coincides with
\[
a\big(\lambda x-\lambda^{-1}y\big)^n
+\sum_{k=0}^{n-1}\big(\lambda x-\lambda^{-1}y\big)^k\big(b_k\lambda^{n-k}+c_k\lambda^{k-n}\big).
\]
Then the condition $Q[u^\lambda]=0$ implies that $a=0$ and $b_k=c_k=0$, $k=0,\dots,n-1$.

In other words, different elements of~$\tilde{\Sigma}^{[n]}$ belong to different cosets of~$\Sigma^{\rm triv}$ in~$\Sigma$,
which are elements of~$\Sigma^{\rm q}$.
Moreover, the order of each of these cosets is equal to~$n$, and $\dim\tilde{\Sigma}^{[n]}=2n+1$.
Then Lemma~\ref{lemma:IDFM:DimLambda} implies
that the space~$\tilde{\Sigma}^{[n]}$ is canonically isomorphic to the space~$\Lambda^{[n]}=\Sigma^{[n]}$.
\end{proof}

It follows from Lemma~\ref{lemma:IDFM:SigmaIsomorphism} that
$\Sigma^{\rm q}=\Lambda^{\rm q}\lsemioplus\Sigma^{-\infty}\simeq\tilde\Sigma^{\rm q}=\tilde\Lambda^{\rm q}\lsemioplus\tilde\Sigma^{-\infty}$, where
\begin{gather*}
%\tilde\Sigma^{\rm q}:=\langle\,(\mathrm J^ku)\p_u,(\mathrm J^k\mathrm D_x^lu)\p_u,(\mathrm J^k\mathrm D_y^lu)\p_u,k\in\mathbb N_0, l\in\mathbb N, f(x,y)\p_u, f\in\mathcal L\rangle.
\Lambda^{\rm q}\simeq\tilde\Lambda^{\rm q}
:=\big\langle(\mathrm J^ku)\p_u,\,(\mathrm J^k\mathrm D_x^lu)\p_u,\,(\mathrm J^k\mathrm D_y^lu)\p_u,\,k\in\mathbb N_0,\,l\in\mathbb N\big\rangle,
\\
\Sigma^{-\infty}\simeq\tilde\Sigma^{-\infty}:=\hat\Sigma^{-\infty}=\big\{f(x,y)\p_u\mid f\in\mathcal L\big\},
\end{gather*}
and all the above isomorphisms are natural as related to quotient spaces.
They become natural isomorphisms related to quotient Lie algebras
if we define the Lie bracket on the space~$\tilde\Sigma^{\rm q}$
as the Lie bracket of generalized vector fields,
where mixed derivatives arising due to the action of the operators~$\mathrm D_x$ and~$\mathrm D_y$ not involved in~$\mathrm J$
should be substituted in view of the equation~$\mathcal L$ and its differential consequences.

The essential Lie invariance algebra~$\mathfrak g^{\rm ess}$ of the equation~$\mathcal L$
is spanned by the vector fields~$\p_x$, $\p_y$, $x\p_x-y\p_y$ and $u\p_u$, cf.~\cite{FushchichNikitin1994}.
It can be identified with the quotient~$\mathfrak g/\tilde\Sigma^{-\infty}$ of the Lie invariance algebra~$\mathfrak g$ of~$\mathcal L$
with respect to the abelian ideal~$\tilde\Sigma^{-\infty}$
corresponding to the linear superposition of solutions of~$\mathcal L$.
Thus, the algebra~$\mathfrak g^{\rm ess}$ is isomorphic to the direct sum of the pseudo-Euclidean algebra~$\mathfrak e(1,1)$
(the Poincar\'e algebra~$\mathfrak p(1,1)$ in another terminology
or the algebra~${\rm g}^{-1}_{3.4}$ in Mubarakzyanov's classification of low-dimensional Lie algebras~\cite{Mubarakzyanov1963a})
and the one-dimensional (abelian) algebra~$\mathfrak a_1$,
$\mathfrak g^{\rm ess}\simeq \mathfrak e(1,1)\oplus \mathfrak a_1$.
Note also that~$\mathfrak g^{\rm ess}\simeq\Lambda^1\simeq\Sigma^1/\Sigma^{-\infty}$.
Let~$\phi\colon\mathfrak g^{\rm ess}\to\mathfrak e(1,1)\oplus \mathfrak a_1$
be the isomorphism with $\phi(u\p_u)=e_0$, $\phi(\p_x)=e_1$, $\phi(\p_y)=e_2$ and~$\phi(x\p_x-y\p_y)=e_3$,
where $\langle e_0\rangle=\mathfrak a_1$ and the basis $(e_1,e_2,e_3)$ of~$\mathfrak e(1,1)$
is related to the standard basis $(\tilde e_1,\tilde e_2,\tilde e_3)$ by $\tilde e_1=e_1+e_2$, $\tilde e_2=e_1-e_2$, $\tilde e_3=e_3$.
The canonical commutation relations of $\mathfrak e(1,1)$ are
$[\tilde e_1,\tilde e_2]=0$, $[\tilde e_1,\tilde e_3]=\tilde e_2$ and $[\tilde e_2,\tilde e_3]=\tilde e_1$,
which take, in the basis $(e_1,e_2,e_3)$, the form $[e_1,e_2]=0$, $[e_1,e_3]=e_1$ and $[e_2,e_3]=-e_2$.
Thus, the universal enveloping algebra~$\mathfrak U(\mathfrak g^{\rm ess})$ of the algebra~$\mathfrak g^{\rm ess}$ is isomorphic to the quotient
of the tensor algebra~$\mathrm T(\mathfrak e(1,1)\oplus \mathfrak a_1)$ by the two-sided ideal~$I$ generated by
$e_1\otimes e_2-e_2\otimes e_1$, $e_1\otimes e_3-e_3\otimes e_1-e_1$, $e_2\otimes e_3-e_3\otimes e_2+e_2$,
$e_0\otimes e_i-e_i\otimes e_0$, $i=1,2,3$.

\begin{theorem}
The quotient algebra~$\Sigma^{\rm q}$ of generalized symmetries of the Klein--Gordon equation~$\mathcal L$
is naturally isomorphic to the algebra~$\tilde\Sigma^{\rm q}$, which is the semidirect sum of the algebra
\[
\tilde\Lambda^{\rm q}=\big\langle(\mathrm J^ku)\p_u,\,(\mathrm J^k\mathrm D_x^lu)\p_u,\,(\mathrm J^k\mathrm D_y^lu)\p_u,\,k\in\mathbb N_0,\,l\in\mathbb N\big\rangle
\simeq\mathfrak U\big(\mathfrak e(1,1)\oplus \mathfrak a_1\big)/\mathfrak I
\]
with the abelian algebra $\tilde\Sigma^{-\infty}=\{f(x,y)\p_u\mid f\in\mathcal L\}$.
Here~$\mathfrak I$ is the two-sided ideal of the universal enveloping algebra $\mathfrak U(\mathfrak e(1,1)\oplus\mathfrak a_1)$
that is generated by the cosets $e_1\otimes e_2-e_0+I$ and $e_0\otimes e_j-e_j+I$, $j=0,1,2,3$.
\end{theorem}

Moreover, for each $Q\in\tilde\Lambda^{\rm q}$ we denote by~$\mathfrak Q$ the linear operator in total derivatives with
coefficients depending on~$x$ and~$y$ that is associated with~$Q$, $Q[u]=\mathfrak Qu$. In this terminology
the operators~$1$, $\mathrm D_x$, $\mathrm D_y$ and~$\mathrm J$ are associated with the evolutionary forms of the
Lie symmetries~$u\p_u$, $-\p_x$, $-\p_y$ and~$y\p_y-x\p_x$ of the Klein--Gordon equation~$\mathcal L$, respectively.
Note that $\mathfrak{LQ=QL}$ for any $Q\in\tilde\Lambda^{\rm q}$.

\noprint{% Another formatting
\begin{corollary}
$\hat\Sigma^{\rm q}=\hat\Lambda^{\rm q}\lsemioplus\hat\Sigma^{-\infty}$, where
$\hat\Lambda^{\rm q}:=\langle\,(\mathscr J^ku)\p_u,\,(\mathscr J^k\mathscr D_x^lu)\p_u,\,(\mathscr J^k\mathscr D_y^lu)\p_u,
\,k\in\mathbb N_0,\,l\in\mathbb N\rangle\simeq\tilde\Lambda^{\rm q}$,
$\hat\Sigma^{-\infty}=\tilde\Sigma^{-\infty}:=\{f(x,y)\p_u\mid f\in\mathcal L\}$,
and $\mathscr J:=x\mathscr D_x-y\mathscr D_y$.
\end{corollary}
}

\begin{corollary}
$\hat\Sigma^{\rm q}=\hat\Lambda^{\rm q}\lsemioplus\hat\Sigma^{-\infty}$, where
\begin{gather*}
\hat\Lambda^{\rm q}:=\big\langle(\mathscr J^ku)\p_u,\,(\mathscr J^k\mathscr D_x^lu)\p_u,\,(\mathscr J^k\mathscr D_y^lu)\p_u,
\,k\in\mathbb N_0,\,l\in\mathbb N\big\rangle\simeq\tilde\Lambda^{\rm q},
\\
\hat\Sigma^{-\infty}=\tilde\Sigma^{-\infty}:=\big\{f(x,y)\p_u\mid f\in\mathcal L\big\},
\end{gather*}
and $\mathscr J:=x\mathscr D_x-y\mathscr D_y$.
\end{corollary}

\section{Variational symmetries}\label{KG:sec:VarSyms}

The (1+1)-dimensional Klein--Gordon equation~$\mathcal L$ is
the Euler--Lagrange equation for the Lagrangian $\mathrm L=-(u_xu_y+u^2)/2$.
Therefore, the spaces $\Sigma$, $\Sigma^{\rm triv}$ and $\Sigma^{\rm q}$
respectively coincide with their counterparts for cosymmetries.
Moreover, in view of Noether's theorem~\cite[Theorem~5.58]{Olver1993}
a~differential function is a conservation-law characteristic of~$\mathcal L$ if and only
if it is the characteristic of a (generalized) variational symmetry of~$\mathrm L$.

Since a generalized vector field is a variational symmetry of a Lagrangian
if and only if its evolutionary representative is~\cite[Proposition~5.32]{Olver1993},
we work only with evolutionary representatives of variational symmetries.
Denote by $\Upsilon$, $\Upsilon^{\rm triv}$ and $\Upsilon^{\rm q}$
the algebra (of evolutionary representatives) of variational symmetries of the Lagrangian~$\mathrm L$,
its subalgebra of trivial variational symmetries and
the quotient algebra of variational symmetries of this Lagrangian, i.e.,
$\Upsilon\subset\Sigma$,
$\Upsilon^{\rm triv}:=\Upsilon\cap\Sigma^{\rm triv}$ and
$\Upsilon^{\rm q}:=\Upsilon/\Upsilon^{\rm triv}$.
In contrast to $\Sigma^{\rm triv}$, the algebra~$\Upsilon^{\rm triv}$
does not consist of all generalized vector fields in the evolutionary form
whose characteristics vanish on solutions of~$\mathcal L$.
This is why one should carefully use
reductions of generalized symmetries by excluding derivatives in view of~$\mathcal L$
when working with variational symmetries, the space of which may not be closed with respect to such a reduction.
We also define the subspace of variational symmetries of order less than or equal to~$n$,
\[
\Upsilon^n=\big\{[Q]\in\Upsilon^{\rm q}\mid\exists\,\eta[u]\p_u\in[Q]\colon\ord\eta[u]\leqslant n\big\}, \quad n\in\mathbb N_0\cup\{-\infty\},
\]
and denote $\Upsilon^{[n]}=\Upsilon^n/\Upsilon^{n-1}$ for $n\in\mathbb N$,
$\Upsilon^{[0]}=\Upsilon^0/\Upsilon^{-\infty}$ and $\Upsilon^{[-\infty]}:=\Upsilon^{-\infty}$.
The space $\Upsilon^{[n]}$ can be interpreted as the space of $n$th order variational symmetries of~$\mathrm L$, $n\in\mathbb N_0\cup\{-\infty\}$.

\begin{lemma}\label{KG:lem:VarSyms}
If a linear generalized symmetry $Q\in\Lambda$ of the Klein--Gordon equation~$\mathcal L$
is a variational symmetry of the Lagrangian~$\mathrm L$, then $\ord Q\in2\mathbb N_0+1$.
\end{lemma}

\begin{proof}
In order for a generalized vector field~$Q$ in~$\Lambda$ to be a variational symmetry of~$\mathrm L$,
its characteristic~$\mathfrak Qu$ has to satisfy the criterion~\cite[Proposition~5.49]{Olver1993}
\[
\mathsf D_{\mathfrak Qu}^\dag(\mathfrak Lu)+\mathsf D^\dag_{\mathfrak Lu}(\mathfrak Qu)
=(\mathfrak Q^\dag\mathfrak L+\mathfrak L^\dag\mathfrak Q)u=0
\]
on the entire infinite-order jet space $\mathrm J^\infty(x,y|u)$.
Here the operator in total derivatives $\mathfrak Q$ corresponds to~$Q$,
$\mathfrak L$ is the operator in total derivatives that is associated with the equation~$\mathcal L$,
$\mathfrak L=\mathrm D_x\mathrm D_y-1$,
a constant summand in a differential operator denotes the multiplication operator by this constant,
$\mathsf D_F$ denotes the Fr\'echet derivative of a differential function~$F$,
and $\mathfrak B^\dag$ denotes the formally adjoint to a differential operator~$\mathfrak B$.
Hence we have the operator equality $\mathfrak Q^\dag\mathfrak L+\mathfrak L^\dag\mathfrak Q=0$.
Since the equation~$\mathcal L$ is the Euler--Lagrange equation of a Lagrangian,
the operator~$\mathfrak L$ is formally self-adjoint, $\mathfrak L^\dag=\mathfrak L$.
If $\ord Q$ were even, then the principal symbol of the left-hand side
of the operator equality $\mathfrak Q^\dag\mathfrak L+\mathfrak L^\dag\mathfrak Q=0$
would be equal to the product of the principal symbols of~$\mathfrak Q$ and~$\mathfrak L$ multiplied by two,
and hence this left-hand side could not be equal to zero.
Therefore, $\ord Q$ is odd.
\end{proof}

\begin{corollary}\label{KG:cor:VarSyms}
A linear generalized symmetry $Q\in\tilde\Lambda^{\rm q}$ of the Klein--Gordon equation~$\mathcal L$
is a variational symmetry of the Lagrangian~$\mathrm L$ if and only if
the corresponding operator~$\mathfrak Q$ is formally skew-adjoint, $\mathfrak Q^\dag=-\mathfrak Q$.
\end{corollary}

\begin{proof}
For $Q\in\tilde\Lambda^{\rm q}$, the operators~$\mathfrak L$ and~$\mathfrak Q$ commute,  $\mathfrak{LQ=QL}$.
This implies
\[
0=\mathfrak Q^\dag\mathfrak L+\mathfrak L^\dag\mathfrak Q=\mathfrak Q^\dag \mathfrak L+\mathfrak L\mathfrak Q=\mathfrak Q^\dag\mathfrak L+\mathfrak Q\mathfrak L=(\mathfrak Q^\dag+\mathfrak Q)\mathfrak L,
\]
and, therefore, $\mathfrak Q^\dag+\mathfrak Q=0$, meaning the desired $\mathfrak Q^\dag=-\mathfrak Q$.
Turning all implications around completes the proof.
\end{proof}

\begin{remark*}
A thorough inspection of the proof of Lemma~\ref{KG:lem:VarSyms} shows
that the same assertion holds for linear variational symmetries of any Lagrangian of one dependent variable 
whose Euler--Lagrange equation is linear.
The assertion analogous to Corollary~\ref{KG:cor:VarSyms} additionally needs commuting
differential operators associated with these symmetries and with the Euler--Lagrange equation.
\end{remark*}

We change the basis of the algebra~$\tilde\Lambda^{\rm q}$ to
$\big((\mathfrak Q_{kl}u)\p_u,\,k,l\in\mathbb N_0,\,(\bar{\mathfrak Q}_{kl}u)\p_u,\,k\in\mathbb N_0,\,l\in\mathbb N\big)$,
where basis' elements are respectively associated with the operators
\begin{gather}\label{eq:IDFM:BasisLambda}
\mathfrak Q_{kl}=\left(\mathrm J+\frac l2\right)^k\mathrm D_x^l, \quad k,l\in\mathbb N_0,\qquad
\bar{\mathfrak Q}_{kl}=\left(\mathrm J-\frac l2\right)^k\mathrm D_y^l, \quad k\in\mathbb N_0,\quad l\in\mathbb N.
\end{gather}
The algebra~$\tilde\Lambda^{\rm q}$ is decomposed into the direct sum of two subspaces,
$\tilde\Lambda^{\rm q}=\tilde\Lambda^{\rm q}_-\dotplus\tilde\Lambda^{\rm q}_+$, where~$\tilde\Lambda^{\rm q}_-$
(resp.\ $\tilde\Lambda^{\rm q}_+$) is the subspace of elements in~$\tilde\Lambda^{\rm q}$
associated with formally skew-adjoint (resp.\ self-adjoint) operators.
Since
\[
\mathrm D_x^\dag=-\mathrm D_x, \quad
\mathrm D_y^\dag=-\mathrm D_y,\quad
\mathrm J^\dag=-\mathrm J,\quad
\mathrm D_x\mathrm J=(\mathrm J+1)\mathrm D_x, \quad
\mathrm D_y\mathrm J=(\mathrm J-1)\mathrm D_y,
\]
we have
\[
\mathfrak Q_{kl}^\dag=(-\mathrm D_x)^l\left(-\mathrm J+\frac l2\right)^k=(-1)^{k+l}\mathrm D_x^l\left(\mathrm J-\frac l2\right)^k=(-1)^{k+l}\mathfrak Q_{kl}
\]
and similarly $\bar{\mathfrak Q}^\dag_{kl}=(-1)^{k+l}\bar{\mathfrak Q}_{kl}$.
Therefore, the generalized vector fields corresponding to the operators~\eqref{eq:IDFM:BasisLambda}
with odd (resp.\ even) values of~$k+l$ constitute
a basis of the space~$\tilde\Lambda^{\rm q}_-$ (resp.\ $\tilde\Lambda^{\rm q}_+$),
\begin{gather*}
\tilde\Lambda^{\rm q}_-=\langle\,(\mathfrak Q_{k'0}u)\p_u,\,k'\in2\mathbb N_0+1,\
(\mathfrak Q_{kl}u)\p_u,\,(\bar{\mathfrak Q}_{kl}u)\p_u,\, k\in\mathbb N_0,\,l\in\mathbb N,\,k+l\in2\mathbb N_0+1 \rangle,
\\
\tilde\Lambda^{\rm q}_+=\langle\,(\mathfrak Q_{k'0}u)\p_u,\,k'\in2\mathbb N_0,\
(\mathfrak Q_{kl}u)\p_u,\,(\bar{\mathfrak Q}_{kl}u)\p_u,\,k\in\mathbb N_0,\,l\in\mathbb N,\,k+l\in2\mathbb N_0 \rangle.
\end{gather*}

\begin{theorem}\label{KG:thm:VarSyms}
The quotient algebra~$\Upsilon^{\rm q}$ of variational symmetries of the Lagrangian~$\mathrm L$
is naturally isomorphic to the algebra $\tilde\Upsilon^{\rm q}=\tilde\Lambda^{\rm q}_-\lsemioplus\tilde\Sigma^{-\infty}$.
\end{theorem}

\begin{proof}
We revert to the coordinates $(x_0,x_1,u)$
and solve the equation~$\mathcal L$ with respect to the derivative $\p^2u/\p x_0^2$, $\p^2u/\p x_0^2=\p^2u/\p x_1^2-u$.
This gives a representation of~$\mathcal L$ in the (extended) Kovalevskaya form.
Lemma~3 in~\cite{MartinezAlonso1979} (which was given in~\cite{Olver1993} as Lemma~4.28)
reformulated for Euler--Lagrange equations in terms of variational symmetries of corresponding Lagrangians implies
that for an arbitrary generalized vector field~$Q$ in~$\Upsilon$,
the corresponding element~$[Q]_{\rm var}$ of~$\Upsilon^{\rm q}$ contains,
as the coset~$Q+\Upsilon^{\rm triv}$ in~$\Upsilon$,
a generalized vector field~$\breve Q$ in the reduced form that is obtained by excluding all derivatives of~$u$
with more than one differentiation with respect to~$x_0$ in view of~$\mathcal L$.
Moreover, $\breve Q$ is the only generalized vector field in the above reduced form
that belongs to the coset~$Q+\Upsilon^{\rm triv}$ in~$\Upsilon$.
It is also the only generalized vector field in the above reduced form
that belongs to the coset~$Q+\Sigma^{\rm triv}$ in~$\Sigma$.
The coset~$Q+\Sigma^{\rm triv}$ necessarily contains
exactly one element of $\tilde\Sigma^{\rm q}=\tilde\Lambda^{\rm q}\lsemioplus\tilde\Sigma^{-\infty}$,
which we denote by~$\tilde Q$.
Note that the used coordinate change preserves the linearity of elements of~$\Lambda$.
Therefore, $\breve Q$ is the reduced form of~$\tilde Q$, and hence $\breve Q\in\Lambda\lsemioplus\tilde\Sigma^{-\infty}$.
Now we can revert to the coordinates $(x,y,u)$.

For any linear system of differential equations,
characteristics of its Lie symmetries associated with the linear superposition of solutions are
conservation-law characteristics of this system.
Therefore, $\tilde\Sigma^{-\infty}\subset\Upsilon$.
Since different elements in~$\tilde\Sigma^{-\infty}$ belong to different elements in the quotient space~$\Upsilon^{\rm q}$
as cosets of~$\Upsilon^{\rm triv}$ in~$\Upsilon$,
and $\ord[Q]=-\infty$ for each $Q\in\tilde\Sigma^{-\infty}$,
the algebra~$\tilde\Sigma^{-\infty}$ is naturally isomorphic to~$\Upsilon^{[-\infty]}$.

By $\smash{\tilde\Lambda^{[n]}_-}$ we denote the subspace of~$\tilde\Lambda^{\rm q}_-$
that is spanned by basis elements of~$\tilde\Lambda^{\rm q}_-$ of order~$n$.
We have $\smash{\tilde\Lambda^{[n]}_-}=\{0\}$ for even~$n$, and if~$n$ is odd, then
\begin{gather*}
\smash{\tilde\Lambda^{[n]}_-}=\big\langle(\mathfrak Q_{n0}u)\p_u,\,(\mathfrak Q_{k,n-k}u)\p_u,\,(\bar{\mathfrak Q}_{k,n-k}u)\p_u,\,k=0,\dots,n-1\big\rangle.
\end{gather*}

Lemma~\ref{KG:lem:VarSyms} implies that if $Q\in\Lambda\cap\Upsilon$, then $\ord Q$ is odd.
Therefore, $\dim\Upsilon^{[n]}=0=\dim\smash{\tilde\Lambda^{[n]}_-}$ for even~$n$.
For odd~$n$, $\dim\Upsilon^{[n]}\leqslant\dim\Sigma^{[n]}=\dim\tilde\Sigma^{[n]}=\dim\tilde\Lambda^{[n]}_-<+\infty$.
On the other hand, $\smash{\tilde\Lambda^{[n]}_-}\subset\Upsilon$, and $\ord[Q]=n$ for each nonzero $Q\in\smash{\tilde\Lambda^{[n]}_-}$.
Hence different elements in~$\smash{\tilde\Lambda^{[n]}_-}$ belong to cosets of~$\Upsilon^{\rm triv}$ in~$\Upsilon$
that are elements of~$\Upsilon^n$ and belong to different cosets of~$\Upsilon^{n-1}$ in~$\Upsilon^n$.
Recall that the latter cosets are considered as elements of the twice quotient space~$\Upsilon^{[n]}$.
This implies that $\dim\smash{\tilde\Lambda^{[n]}_-}\leqslant\dim\Upsilon^{[n]}$.
In total, for odd~$n$ this gives that $\dim\smash{\tilde\Lambda^{[n]}_-}=\dim\Upsilon^{[n]}$,
and the subspace $\smash{\tilde\Lambda^{[n]}_-}$ of $\Upsilon$
is naturally isomorphic to the space~$\Upsilon^{[n]}$ via taking quotients twice.
Therefore, the subspace~$\Upsilon^n$ of~$\Upsilon^{\rm q}$ is naturally isomorphic to
the subspace $\tilde\Sigma^{-\infty}\dotplus\smash{\tilde\Lambda^{[0]}_-}\dotplus\dots\dotplus\smash{\tilde\Lambda^{[n]}_-}$ of~$\Upsilon$.
Then the algebra~$\Upsilon^{\rm q}$ is naturally isomorphic
to the algebra $\tilde\Upsilon^{\rm q}=\tilde\Lambda^{\rm q}_-\lsemioplus\tilde\Sigma^{-\infty}$.
Here the Lie bracket on~$\tilde\Upsilon^{\rm q}$ is defined similarly to the Lie bracket on~$\tilde\Sigma^{\rm q}$, i.e.,
as the Lie bracket of generalized vector fields,
where mixed derivatives arising due to the action of~$\mathrm D_x$ and~$\mathrm D_y$ not involved in~$\mathrm J$
should be substituted in view of the equation~$\mathcal L$ and its differential consequences.
\end{proof}

\begin{remark}\label{KG:rem:VarSym}
Cosets of~$\Upsilon^{\rm triv}$ in~$\Upsilon$
do not necessarily intersect the algebra~$\hat\Sigma^{\rm q}$,
i.e., they do not have canonical representatives in the evolutionary form
reduced on solutions of the equation~$\mathcal L$.
For example, the reduced counterpart $(\mathscr J^3u)\p_u$
of the variational symmetry $(\mathfrak Q_{30}u)\p_u=(\mathrm J^3u)\p_u$
%$=(x^3u_{xxx}-3x^2yu_{xxy}+3xy^2u_{xyy}+y^3u_{yyy}+3x^2u_{xx}-3y^2u_{yy}+xu_x-yu_y)\p_u$
of~$\mathrm L$
is not a variational symmetry of~$\mathrm L$
since the difference $(\mathscr J^3u)\p_u-(\mathrm J^3u)\p_u=3xy\mathrm J(u_{xy}-u)\p_u$ is not.
Recall that $\mathscr J:=x\mathscr D_x-y\mathscr D_y$.
In other words, the reduced evolutionary form of generalized symmetries of the Klein--Gordon equation~$\mathcal L$
is not appropriate in the course of the study of variational symmetries of~$\mathrm L$.
\end{remark}

\section{Conservation laws}\label{KG:sec:ConsLaws}

For each element in a set spanning the space~$\tilde\Upsilon^{\rm q}$,
we construct a conserved current of the corresponding conservation law.
Moreover, these conserved currents are of the simplest form
and of minimal order among equivalent conserved currents,
that is, their orders coincide with the orders of conservation laws containing them.
In the course of this construction,
we multiply the differential function~$\mathfrak Lu$
by the characteristic of a variational symmetry of~$\mathrm L$
and rewrite, ``integrating by parts'', this expression
in the form of a total divergence of a tuple of differential functions,
which is nothing else but a conserved current of~$\mathcal L$.

Thus, for any element~$f(x,y)\p_u$ of~$\tilde\Sigma^{-\infty}$,
the function~$f=f(x,y)$ is a solution of~$\mathcal L$, and we have
$f\mathfrak Lu=\mathrm D_x(fu_y)+\mathrm D_y(-f_xu)=\mathrm D_x(-f_yu)+\mathrm D_y(fu_x)$,
which yields the equivalent first-order conserved currents
\[
{\rm C}^0_f=(fu_y,-f_xu)\quad\mbox{and}\quad \bar{\rm C}^0_f=(-f_yu,fu_x).
\]

We can use a similar trick to derive a conserved current of~$\mathcal L$
using any $Q=(\mathfrak Qu)\p_u\in\tilde\Lambda^{\rm q}$. We get
\begin{gather*}
\mathrm D_x(-u\mathrm D_y\mathfrak Qu)+\mathrm D_y(u_x\mathfrak Qu)
=u_{xy}\mathfrak Qu-u\mathrm D_x\mathrm D_y\mathfrak Qu
=(\mathfrak Qu)\mathfrak Lu-u\mathfrak L\mathfrak Qu
\\ \qquad
=(\mathfrak Qu)\mathfrak Lu-u\mathfrak Q\mathfrak Lu
=(\mathfrak Qu-\mathfrak Q^\dag u)\mathfrak Lu+(\mathfrak Q^\dag u)\mathfrak Lu-u\mathfrak Q\mathfrak Lu.
\end{gather*}
Here we take into account that $\mathfrak L\mathfrak Q=\mathfrak Q\mathfrak L$ for $Q\in\tilde\Lambda^{\rm q}$.
The Lagrange identity (also called generalized Green's formula \cite[Section~12]{Zharinov1992})
implies that the differential function $(\mathfrak Q^\dag u)\mathfrak Lu-u\mathfrak Q\mathfrak Lu$
is the total divergence of a pair of differential functions
bilinearly depending on the tuples of total derivatives of $u$ and~$\mathfrak Lu$;
cf.\ \cite[Proposition~A.4]{Zharinov1992},
i.e., it is the total divergence
of a trivial conserved current of the equation~$\mathcal L$.
Therefore, $(\mathfrak Q-\mathfrak Q^\dag)u$ is a characteristic of the conservation law of~$\mathcal L$
that contains the conserved current
\[
\tilde{\rm C}_{\mathfrak Q}=(-u\mathrm D_y\mathfrak Qu,\,u_x\mathfrak Qu).
\]
For any $Q\in\tilde\Lambda^{\rm q}_+$, we have $\mathfrak Q^\dag=\mathfrak Q$,
i.e., the corresponding conservation law is zero.
For any $Q\in\tilde\Lambda^{\rm q}_-$, we have $\mathfrak Q^\dag=-\mathfrak Q$ and thus obtain
the characteristic $2\mathfrak Qu$ of a nonzero conservation law of~$\mathcal L$.
Running $Q$ through the basis of~$\tilde\Lambda^{\rm q}_-$ gives conservation laws
that are linearly independent since their characteristics are.
In view of Theorem~\ref{KG:thm:VarSyms}, these conservation laws
jointly with those containing conserved currents ${\rm C}^0_f$, $f\in\mathcal L$,
span the entire space of conservation laws of~$\mathcal L$.

\begin{proposition}
The space of conservation laws of the $(1+1)$-dimensional Klein--Gordon equation~$\mathcal L$
is naturally isomorphic to the space spanned by the conserved currents
${\rm C}^0_f$ and~$\tilde{\rm C}_{\mathfrak Q}$,
where the parameter function $f=f(x,y)$ runs through the solution set of~$\mathcal L$,
and the operator~$\mathfrak Q$ runs through the basis of~$\tilde\Lambda^{\rm q}_-$,
\[
(\,\mathfrak Q_{k'0},\,k'\in2\mathbb N_0+1,\
\mathfrak Q_{kl},\,\bar{\mathfrak Q}_{kl},\, k\in\mathbb N_0,\,l\in\mathbb N,\,k+l\in2\mathbb N_0+1\,).
\]
\end{proposition}

\begin{corollary}
Under the action of generalized symmetries of the $(1+1)$-dimensional Klein--Gordon equation~$\mathcal L$
on the space of conservation laws of this equation,
a generating set of conservation laws of~$\mathcal L$
is constituted by the single conservation law containing the conserved current \[(-u^2,u_x^2).\]
\end{corollary}

\begin{proof}
The actions of the generalized symmetries $\frac12f_y\p_u$ and $\frac12(\mathrm D_y\mathfrak Qu)\p_u$
on the conserved current $(-u^2,u_x^2)$ give
the conserved currents $\bar{\rm C}^0_f=(-f_yu,fu_x)$ and $(-u\mathrm D_y\mathfrak Qu,u_x\mathrm D_x\mathrm D_y\mathfrak Qu)$,
which are equivalent to~${\rm C}^0_f$ and~$\tilde{\rm C}_{\mathfrak Q}$, respectively.
\end{proof}

The order of the conserved current~$\tilde{\rm C}_{\mathfrak Q}$
is greater than the order of the corresponding conservation law.
This is why we compute a conserved current of minimal order with characteristic~$\mathfrak Qu$,
where the generalized vector field $(\mathfrak Qu)\p_u$ runs through
the chosen basis elements~$(\mathfrak Q_{kl}u)\p_u$ of~$\tilde\Lambda^{\rm q}_-$, for each of which $k+l$ is odd.
We consider two cases, when $k$ is odd and when $k$ is even.

In the first case, we denote~$k'=(k-1)/2$ and~$l'=l/2$.
Note that $\mathrm J=\mathrm D_x\circ x-\mathrm D_y\circ y$.
Hence $\mathfrak Q_{kl}=\mathrm D_x^{l'}\mathrm J^k\mathrm D_x^{l'}$ and
\begin{gather*}
\begin{split}
(\mathfrak Q_{kl}u)\mathfrak Lu
\noprint{
=\mathrm D_x\sum\limits_{l''=0}^{l'-1}(-1)^{l''}\left(\mathrm D_x^{l'-l''-1}\mathrm J^k\mathrm D_x^{l'}u\right)\mathrm D_x^{l''}\mathfrak Lu
+(-1)^{l'}\left(\mathrm J^k\mathrm D_x^{l'}u\right)\mathrm D_x^{l'}\mathfrak Lu
\\ \qquad
=\mathrm D_x\sum\limits_{l''=0}^{l'-1}(-1)^{l''}\left(\mathrm D_x^{l'-l''-1}\mathrm J^k\mathrm D_x^{l'}u\right)\mathrm D_x^{l''}\mathfrak Lu
+\mathrm J\sum\limits_{k''=0}^{k'-1}(-1)^{l'+k''}\left(\mathrm J^{2k'-k''}\mathrm D_x^{l'}u\right)\mathrm J^{k''}\mathrm D_x^{l'}\mathfrak Lu
\\ \qquad\quad{}
+(-1)^{l'+k'}\left(\mathrm J^{k'+1}\mathrm D_x^{l'}u\right)\mathrm J^{k'}\mathrm D_x^{l'}\mathfrak Lu
\\ \qquad
=\mathrm D_x\sum\limits_{l''=0}^{l'-1}(-1)^{l''}\left(\mathrm D_x^{l'-l''-1}\mathrm J^k\mathrm D_x^{l'}u\right)\mathrm D_x^{l''}\mathfrak Lu
+\mathrm J\sum\limits_{k''=0}^{k'-1}(-1)^{l'+k''}\left(\mathrm J^{2k'-k''}\mathrm D_x^{l'}u\right)\mathrm J^{k''}\mathrm D_x^{l'}\mathfrak Lu
\\ \qquad\quad{}
+(-1)^{l'+k'}\left(\mathrm J^{k'+1}\mathrm D_x^{l'}u\right)\mathrm J^{k'}\mathrm D_x^{l'}(\mathrm D_x\mathrm D_y-1)u
\\ \qquad
}
={}&\mathrm D_x\sum\limits_{l''=0}^{l'-1}(-1)^{l''}\left(\mathrm D_x^{l'-l''-1}\mathrm J^k\mathrm D_x^{l'}u\right)\mathrm D_x^{l''}\mathfrak Lu
\\[.5ex]
&+\mathrm J\sum\limits_{k''=0}^{k'-1}(-1)^{l'+k''}\left(\mathrm J^{2k'-k''}\mathrm D_x^{l'}u\right)\mathrm J^{k''}\mathrm D_x^{l'}\mathfrak Lu
\\[.5ex]
&+\frac{(-1)^{l'+k'}}2\left(
x\mathrm D_y(\mathrm D_x\mathrm J^{k'}\mathrm D_x^{l'}u)^2-y\mathrm D_x(\mathrm D_y\mathrm J^{k'}\mathrm D_x^{l'}u)^2
-\mathrm J(\mathrm J^{k'}\mathrm D_x^{l'}u)^2\right),
\end{split}
\end{gather*}
which gives, up to the equivalence of conserved currents of~$\mathcal L$ and their rescaling, the conserved current
\[
{\rm C}^1_{k'l'}=\left(
-y(\mathrm D_y\mathrm J^{k'}\mathrm D_x^{l'}u)^2-x(\mathrm J^{k'}\mathrm D_x^{l'}u)^2,\
 x(\mathrm D_x\mathrm J^{k'}\mathrm D_x^{l'}u)^2+y(\mathrm J^{k'}\mathrm D_x^{l'}u)^2\right)
\]
of order~$k'+l'+1=(k+l+1)/2$, which is minimal for the conserved currents related to the characteristic~$\mathfrak Q_{kl}u$.

If $k$ is even, then~$l$ is odd and we denote~$k'=k/2$ and~$l'=(l-1)/2$.
Hence
\[\mathfrak Q_{kl}
=\mathrm D_x^{l'}(\mathrm J+1/2)^{k'}\mathrm D_x(\mathrm J-1/2)^{k'}\mathrm D_x^{l'}
=\mathrm D_x^{l'+1}(\mathrm J-1/2)^k\mathrm D_x^{l'}
=\mathrm D_x^{l'}(\mathrm J+1/2)^k\mathrm D_x^{l'+1}
\] and
\begin{gather*}
\begin{split}
(\mathfrak Q_{kl}u)\mathfrak Lu
={}&\mathrm D_x\sum\limits_{l''=0}^{l'-1}(-1)^{l''}\left(\mathrm D_x^{l'-l''}\left(\mathrm J-\frac12\right)^k\mathrm D_x^{l'}u\right)\mathrm D_x^{l''}\mathfrak Lu
\\[.5ex]
&+\mathrm J\sum\limits_{k''=0}^{k'-1}(-1)^{l'+k''}\left(\left(\mathrm J+\frac12\right)^{k-k''-1}\mathrm D_x^{l'+1}u\right)
\left(\mathrm J-\frac12\right)^{k''}\mathrm D_x^{l'}\mathfrak Lu
\\[.5ex]
&+\frac{(-1)^{l'+k'}}2\left(
 \mathrm D_y\bigg(\mathrm D_x\left(\mathrm J-\frac12\right)^{k'}\mathrm D_x^{l'}u\bigg)^2
-\mathrm D_x\bigg(\left(\mathrm J-\frac12\right)^{k'}\mathrm D_x^{l'}u\bigg)^2
\right).
\end{split}
\end{gather*}
Up to the equivalence of conserved currents of~$\mathcal L$ and multiplying them by constants, this leads to the conserved current
\[
{\rm C}^2_{k'l'}=\left(
-\bigg(\left(\mathrm J-\frac12\right)^{k'}\mathrm D_x^{l'}u\bigg)^2,\
 \bigg(\mathrm D_x\left(\mathrm J-\frac12\right)^{k'}\mathrm D_x^{l'}u\bigg)^2\right)
\]
of order~$k'+l'+1=(k+l+1)/2$,
which is again minimal for the conserved currents related to the characteristic~$\mathfrak Q_{kl}u$.
Since the permutation of~$x$ and~$y$ is a discrete point symmetry transformation of the equation~$\mathcal L$,
a conserved current associated with the variational symmetry~$(\bar{\mathfrak Q}_{kl}u)\p_u$ of the Lagrangian~$\mathrm L$, for which $k+l$ is odd,
can be constructed by this permutation
either from the conserved current~${\rm C}^1_{k'l'}$ if $k$ is odd
or from the conserved current~${\rm C}^2_{k'l'}$ if $k$ is even,
where again $k'$ and~$l'$ denote the integer parts of $k/2$ and~$l/2$, respectively.
We obtain
\begin{gather*}
\bar{\rm C}^1_{k'l'}=\left(
 y(\mathrm D_y\mathrm J^{k'}\mathrm D_y^{l'}u)^2+x(\mathrm J^{k'}\mathrm D_y^{l'}u)^2,\
-x(\mathrm D_x\mathrm J^{k'}\mathrm D_y^{l'}u)^2-y(\mathrm J^{k'}\mathrm D_y^{l'}u)^2\right),
\\[.5ex]
\bar{\rm C}^2_{k'l'}=\left(
\bigg(\mathrm D_y\left(\mathrm J+\frac12\right)^{k'}\mathrm D_y^{l'}u\bigg)^2,\
-\bigg(\left(\mathrm J+\frac12\right)^{k'}\mathrm D_y^{l'}u\bigg)^2\right).
\end{gather*}

We sum up the above construction of conserved currents as the following assertion.

\begin{theorem}\label{KG:theor:ConserLaws}
The space of conservation laws of the $(1+1)$-dimensional Klein--Gordon equation~$\mathcal L$
is naturally isomorphic to the space spanned by the conserved currents
\[
{\rm C}^1_{k'l'},\ k'\in\mathbb N_0,\ l'\in\mathbb N,\quad
\bar{\rm C}^1_{k'l'},\ {\rm C}^2_{k'l'},\ \bar{\rm C}^2_{k'l'},\ k',l'\in\mathbb N_0,\quad
{\rm C}^0_f,
\]
where the parameter function $f=f(x,y)$ runs through the solution set of~$\mathcal L$.
The order of conserved currents ${\rm C}_{k'l'}$'s is equal to $k'+l'+1$, and
$\ord{\rm C}^0_f=1$.
\end{theorem}

In other words, the conserved currents ${\rm C}^1_{k'l'}$, $k'\in\mathbb N_0$, $l'\in\mathbb N$,
$\bar{\rm C}^1_{k'l'}$, ${\rm C}^2_{k'l'}$, $\bar{\rm C}^2_{k'l'}$, $k',l'\in\mathbb N_0$,
with $k'+l'=n-1$ represent a complete (up to adding lower-order conservation laws)
set of linearly independent $n$th order conservation laws of~$\mathcal L$ if $n\geqslant2$.
The space of first-order conservation laws is spanned by those with conserved currents
$\bar{\rm C}^1_{00}$, ${\rm C}^2_{00}$, $\bar{\rm C}^2_{00}$ and ${\rm C}^0_f$,
where the parameter function $f=f(x,y)$ runs through the solution set of~$\mathcal L$.

\begin{corollary}
Up to adding low-order conservation laws, the Klein--Gordon equation~$\mathcal L$ possesses
$4n-1$ linearly independent conservation laws of order~$n$ if $n\geqslant2$,
 and an infinite number of linearly independent first-order conservation laws.
\end{corollary}

\begin{remark*}
Replacing the operators~$\mathrm D_x$, $\mathrm D_y$ and~$\mathrm J$
by~$\mathscr D_x$, $\mathscr D_y$ and~$\mathscr J$, respectively,
in constructed conserved currents,
we obtain equivalent conserved currents that are reduced in view of the solution set of~$\mathcal L$.
\end{remark*}

\section{Conclusion}\label{KG:sec:Conclusion}

The consideration in the present paper has several interesting aspects, which are worth recalling.
Its main specific feature is
that it is essentially based on the representation~$\mathcal L$: $u_{xy}=u$
of the (1+1)-dimensional Klein--Gordon equation in the light-cone variables,
which cannot be adapted,
in contrast to the representation in the standard spacetime variables,
as an (extended) Kovalevskaya form of this equation.% 
\footnote{%
See~\cite{PopovychBihlo2020} for the definition of the extended Kovalevskaya form of systems of differential equations
and a discussion of significance of this form in the theory of conservation laws. 
Systems of a bit more restrictive form are called normal systems~\cite{MartinezAlonso1979} 
or Cauchy--Kowalevsky systems in a weak sense (resp., pseudo CK systems in short)~\cite{Tsujishita1982}.
}
There are only a few papers in the literature,
where the entire spaces of generalized symmetries and, especially, conservation laws
were computed for (systems of) differential equations  
that are inconvenient for representing in the extended Kovalevskaya form \cite{DuzhinTsujishita1984,Sharomet1989}
or lack such a representation at all 
\cite{AncoPohjanpelto2001,AncoPohjanpelto2003,AncoPohjanpelto2004,AndersonTorre1996,Pohjanpelto2004,PohjanpeltoAnco2008}.
Moreover, in \mbox{\cite{DuzhinTsujishita1984,Sharomet1989}} the least upper bounds for orders of reduced cosymmetries were low,
2 and $-\infty$, respectively,
each equivalence class of cosymmetries contained a conservation-law characteristic,
and the sufficient number of linearly independent conservation laws had been known~\cite{DuzhinTsujishita1984}
or could be easily derived directly~\cite{Sharomet1989}.
This is why employing the equation representations
different from the extended Kovalevskaya form
created no obstacles for selecting conservation-law characteristics among cosymmetries in these papers
although, in general, such a selection may be a nontrivial problem.
Thus, the present paper provides one of a few examples of studying conservation laws of a system of differential equations
that is not in the extended Kovalevskaya form and possesses conservation laws of arbitrarily high order
as well as cosymmetries of arbitrarily high order that are not equivalent to conservation-law characteristics, 
cf.~\cite{AncoPohjanpelto2001,AncoPohjanpelto2003,AncoPohjanpelto2004,AncoThe2005}.
To get around the complication in the course of selecting variational symmetries among generalized ones
for the representation of the Klein--Gordon equation~$\mathcal L$ in the light-cone variables~$x$ and~$y$,
we have temporarily switched to the standard form of the Klein--Gordon equation
for applying the Mart\'inez Alonso lemma~\cite[Lemma~3]{MartinezAlonso1979}.
That the transitions between the standard spacetime and the light-cone variables
preserve the linearity of characteristics of generalized symmetries
allowed us to prove that each nonnegative-order coset of variational symmetries
contains a linear symmetry.
All the other computations were carried out in the light-cone variables.

Despite the above complication,
the representation of the Klein--Gordon equation~$\mathcal L$ in the light-cone variables~$x$ and~$y$
is preferable to the standard one.
The choice of it is paid off by virtue of the facts
that it is more compact and
the differentiations with respect to~$x$ and~$y$ are inverse to each other on solutions of~$\mathcal L$, 
$\mathrm D_x|_{\mathcal L}^{}\circ\mathrm D_y|_{\mathcal L}^{}=\mathrm{id}|_{\mathcal L}^{}$.
The latter enables us to choose the jet coordinates $(t,x,u_k,k\in\mathbb Z)$ on~$\mathcal L^{(\infty)}$,
which are numerated by a single integer.
This simplifies the entire consideration,
including the reduced operators of total derivatives~$\mathscr D_x$ and~$\mathscr D_y$,
the determining equations for generalized symmetries of~$\mathcal L$
and the process of solving thereof.
In contrast to the standard spacetime coordinates,
the use of light-cone variables in the course of confining to the solution set of the Klein--Gordon equation
also allows us to preserve the equality of independent variables, which is intrinsic to this equation.
As a result, both the constructed spaces of canonical representatives
for equivalence classes of generalized symmetries of~$\mathcal L$
admit bases that are symmetrical with respect to~$x$ and~$y$.

The procedure of finding generalized symmetries of~$\mathcal L$ includes the standard techniques
of computing the dimension of the space of reduced generalized symmetries of each finite order
and of generating the necessary amount of linearly independent symmetries by recursion operators.
In fact, for the latter it suffices to use only the recursion operators,
corresponding to the Lie symmetries~$\p_x$, $\p_y$ and~$x\p_x-y\p_y$ of~$\mathcal L$.
To show that the generation produces no trivial symmetries,
we have evaluated the constructed generalized symmetries
on a family of solutions of~$\mathcal L$ parameterized by a nonzero real constant,
see the proof of Lemma~\ref{lemma:IDFM:SigmaIsomorphism}.
From this perspective, the entire algebra~$\tilde\Sigma^{\rm q}$ (resp.\ $\hat\Sigma^{\rm q}$)
of canonical representatives for equivalence classes of generalized symmetries of~$\mathcal L$
is spanned by the generalized vector fields that are related to the linear superposition of solutions of~$\mathcal L$
or generated from the single Lie symmetry~$u\p_u$ of~$\mathcal L$
by means of the recursion operators~$\mathrm D_x$, $\mathrm D_y$ and $\mathrm J$
(resp.\ $\mathscr D_x$, $\mathscr D_y$ and $\mathscr J$).
The algebra~$\hat\Sigma^{\rm q}$ is the collection of generalized symmetries of~$\mathcal L$
reduced on the solution set of~$\mathcal L$, thus being a standard object.
Moreover, the elements of~$\hat\Sigma^{\rm q}$ are represented in a compact form, in particular,
due to the obtained compact representation of the reduced operators of total derivatives~$\mathscr D_x$ and $\mathscr D_y$.
Nevertheless, the algebra~$\hat\Sigma^{\rm q}$ is inappropriate for use
in the description of variational symmetries of the equation~$\mathcal L$, see Remark~\ref{KG:rem:VarSym}.
This is why we have paid a more attention to another collection
of canonical representatives for equivalence classes of generalized symmetries of~$\mathcal L$,
the algebra~$\tilde\Sigma^{\rm q}$,
which does not have the above disadvantage of the algebra~$\hat\Sigma^{\rm q}$.
In order to efficiently single out variational symmetries among elements of the algebra~$\tilde\Sigma^{\rm q}$,
we have made a basis change in this algebra,
so that the subspace of skew-adjoint operators, which are naturally associated with variational symmetries,
is evident in the new basis.

The space of conservation laws of~$\mathcal L$ is expectedly computed using Noether's theorem.
It~is convenient to represent this space as the direct sum of two infinite-dimensional subspaces.
The first subspace is of the kind that is common for linear systems of differential equations.
It~consists of the (first-order) linear conservation laws of~$\mathcal L$.
Such conservation laws are necessarily of order one, and their (reduced) characteristics are of order $-\infty$.
For~$\mathcal L$ as the Euler--Lagrange equation of the Lagrangian~$\mathrm L$,
these characteristics are characteristics of generalized symmetries of order $-\infty$ of~$\mathcal L$,
which constitute the algebra~$\tilde\Sigma^{-\infty}$
and are associated with the linear superposition of solutions of~$\mathcal L$.
The second subspace is specific and is exhausted by the quadratic conservation laws of~$\mathcal L$.
They admit linear characteristics being characteristics of linear variational symmetries from the algebra~$\tilde\Lambda^{\rm q}_-$.
We have derived canonical representatives of two kinds for conserved currents contained in quadratic conservation laws.
The first kind of representatives is uniform for all quadratic conservation laws
and is convenient in the course of the study
how generalized symmetries of the equation~$\mathcal L$ act on its conservation laws.
It was an unexpected result for us that the so huge space of conservation laws of diverse structures is generated,
under the action of generalized symmetries, by a single first-order quadratic conservation law.
We have also computed a conserved current of minimal order for each basis quadratic conservation law.
For computational and presentation reasons, in the course of this computation
we partition the chosen basis of variational symmetries of nonnegative order into four families,
which leads to the associated partition for quadratic conservation laws.
We have constructed conserved currents of minimal order for two of these four families of conservation laws
and then used the permutation of~$x$ and~$y$,
which is a discrete point symmetry transformation of~$\mathcal L$,
to obtain conserved currents of minimal order for the other two families from the constructed ones.

An additional advantage of using the operators~$\mathrm D_x$ and~$\mathrm D_y$
over their rivals~$\mathscr D_x$ and $\mathscr D_y$
is a more clear insight into generalizing results of the present paper
to the multi-dimensional Klein--Gordon equation.
In view of the greater number of independent variables,
it possesses more translations and Lorentz transformations (usual and hyperbolic rotations)
than the equation~$\mathcal L$ does
but the principal structure of the algebra of generalized symmetries
should be similar to that for~$\mathcal L$, 
cf.\ \cite{Eastwood2005,Nikitin1991,NikitinPrylypko1990,ShapovalovShirokov1992}.
The techniques applied in the present paper for singling out variational symmetries
and computing associated conserved currents of minimal order may still be employed
for constructing the entire space of conservation laws of the multi-dimensional Klein--Gordon equation,
including the translation-noninvariant ones, which were not considered in~\cite{Kibble1965,Tsujishita1979}.

\section*{Acknowledgments}

The authors are grateful to the anonymous reviewer for valuable remarks, 
and to Alexander Bihlo, Anatoly Nikitin, Dmytro Popovych and Artur Sergyeyev for useful discussions and interesting comments. 
This research was undertaken, in part, thanks to funding from the Canada Research Chairs program and the NSERC Discovery Grant program.
The research of ROP was supported by the Austrian Science Fund (FWF), projects P25064 and P30233.
ROP is also grateful to the project No.\ CZ.$02.2.69\/0.0/0.0/16\_027/0008521$
``Support of International Mobility of Researchers at SU'' which supports international cooperation.

\footnotesize\frenchspacing	 
%\bibliography{opanasenko}

\end{document}